\documentclass[conference]{IEEEtran}
\usepackage{amsmath}
\usepackage{amssymb}
\usepackage{amsthm}
\usepackage[dvipsnames]{xcolor}
\usepackage{enumitem}
\usepackage{hyperref}
\usepackage{cite}
\usepackage{graphicx}
\usepackage{tikz}
\ifCLASSOPTIONcompsoc\usepackage[caption=false,font=normalsize,labelfont=sf,textfont=sf]{subfig}\else\usepackage[caption=false,font=footnotesize]{subfig}\fi
\newtheorem{lemma}{Lemma}

\newtheorem{theorem}{Theorem}

\theoremstyle{definition}

\hypersetup{
	colorlinks,
	linkcolor={red!50!black},
	citecolor={blue!50!black},
	urlcolor={blue!80!black}
}

\author{Olle Abrahamsson, Danyo Danev and Erik G. Larsson\\ \\
Dept. of Electrical Engineering (ISY), \\Link{\"o}ping University, 58183 Link{\"o}ping, Sweden\\
Email: \{olle.abrahamsson, danyo.danev, erik.g.larsson\}@liu.se}
\title{Structural Balance Considerations for Networks with Preference Orders as Node Attributes}

\makeatletter
\def\footnoterule{\kern-3\p@
  \hrule \@width 2in \kern 2.6\p@} % the \hrule is .4pt high
\makeatother

\newcommand\blfootnote[1]{%
  \begingroup
  \renewcommand\thefootnote{}\footnote{#1}%
  \addtocounter{footnote}{-1}%
  \endgroup
}

\begin{document}
	\maketitle
	
	\begin{abstract}
		We discuss possible definitions of structural balance conditions in a network with preference orderings as node attributes. The main result is that for the case with three alternatives (\(A,B,C\)) we reduce the \((3!)^3 = 216\) possible configurations of triangles to \(10\) equivalence classes, and use these as measures of balance of a triangle towards possible extensions of structural balance theory. Moreover, we derive a general formula for the number of equivalent classes for preferences on \(n\) alternatives. Finally, we analyze a real-world data set and compare its empirical distribution of triangle equivalence classes to a null hypothesis in which preferences are randomly assigned to the nodes.
	\end{abstract}

\blfootnote{This work was supported in part by ELLIIT and the KAW foundation.}
\section{Introduction}\label{sec:intro}

In network science, nodes and edges may be associated with attributes
that encode various properties. An important example of a node
attribute is when each node is assigned a particular \emph{type}.
Various homophily measures, with modularity as the most important
example \cite{newmangirvan04}, are then available to quantify whether edges between
nodes of the same type are more prevalent than edges between nodes of
different types.  Edge attributes may consist of, for example, a
weight that describes how strongly two nodes are connected, or a sign
($+/-$) that determines whether the the relation between two nodes is
friendly or antagonistic. For complete networks with signed edges, a celebrated
result is the structural balance theorem \cite{cartwrightharary}.  This theorem states
that if every triangle in a (complete) network has the signs either $+++$ or
$+--$, then the network can be partitioned into two subnetworks \(\mathcal{A}\) and
\(\mathcal{B}\), such that all edges within \(\mathcal{A}\) are $+$, all edges within \(\mathcal{B}\) are \(+\),
but all edges between \(\mathcal{A}\) and \(\mathcal{B}\) are $-$ (as illustrated in Figure~\ref{fig:structbal}).  The network is
then said to be \emph{balanced}.  The popular interpretation is that
the ``enemy's enemy is a friend'': in a balanced network either all
three nodes in every triangle are friends, or two team up against a
third, common, enemy. Thus, in signed networks balance means polarization, i.e., there are two
camps of friends with mutual antagonism between them.
		\begin{figure}[h]
	\centering
	\begin{tikzpicture}[auto,main_node/.style={circle,fill=blue!20,draw,minimum size=1em,inner sep=3pt]}]
	
	\node[main_node] (1) at (0,0)  {\(v_1\)};
	\node[main_node] (2) at (0, -1.7)  {\(v_2\)};
	\node[main_node] (3) at (-1.7,-3)  {\(v_3\)};
	\node[main_node] (4) at (1.7,- 3){\(v_4\)};

	\draw (1) to node [swap] {\(+\)} (2);
	\draw (1) to node [swap]  {\(-\)} (3);
	\draw (1) to node {\(-\)} (4);
	\draw (2) to node {\(-\)} (3);
	\draw (2) to node [swap] {\(-\)} (4);
	\draw (3) to node [swap] {\(+\)} (4);
	\end{tikzpicture}%
	\caption{A balanced network with subnetworks \(\mathcal{A}\), induced by the nodes \(v_1\) and \(v_2\), and \(\mathcal{B}\), induced by \(v_3\) and \(v_4\). All edges within \(\mathcal{A}\) and  \(\mathcal{B}\) have positive signs, but the edges between \(\mathcal{A}\) and \(\mathcal{B}\) have negative signs. }\label{fig:structbal}
\end{figure}
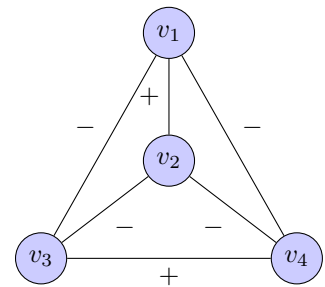

In this paper we are concerned with extensions of the structural balance concept that go beyond signed networks. Specifically we
consider networks where nodes have \emph{preference orderings} as attributes.  A preference ordering is defined in terms of a strict linear order
of alternatives  (e.g., \( A\succ B \succ C\)) and may be interpreted as  an expression of the opinion that \(A\) is preferable to \(B\), which in turn is preferable to \(C\),
 and so on. The extent to which two neighbors' preference orderings differ can be viewed as a measure of agreement or disagreement. In the extreme cases when a pair of neighbors have the exact same preferences or maximally different preferences (e.g., according to some distance metric), it would be quite natural to associate their edge with a \(+\) or \(-\), respectively. In the case of partial disagreement, one could either extend the number of edge weights or argue for a way to project the partial disagreements onto the two signs \(+/-\).  In view of this intuitive connection between preference orderings as node attributes and signed edges,  one might ask a more general question: Can triangles of nodes with preference orderings as node attributes be associated with a notion of balance? We approach, and partially answer this question by enumerating all possible combinations of preference orderings that can appear in a triangle, and categorize them into classes that are equivalent in a certain sense.

\subsection{Related Work}

\subsubsection{Preference Orderings as Node Attributes}

The motivation for studying networks with preference orderings is
clear.  For example, it is reported in the literature that the use of preference orderings
to express opinions is superior to the use of numerical scores \cite{ranknotscore, brill}. In the field of collaborative filtering-based recommender systems, several papers have leveraged on users' preference orderings data \cite{brun10,jones11,liu16}, and in social choice theory, an important problem is how to aggregate preference orderings in a group of people \cite{bajgiran21}, with applications in the design of voting methods and ethical AI systems \cite{bana21}.

Yet, only scattered results are available in the literature on networks with preference orderings as node attributes.  In \cite{brill}, opinion diffusion processes were considered and it was concluded that the outcome of a diffusion process depends both on the structure of the social network (e.g., directed or undirected; cyclic or acyclic) and on the properties of the initial preferences of each agent, for example whether the initial preference profiles satisfy the Condorcet condition or not. In \cite{salehi19} (with results refined in \cite{becirovic17}), preference aggregation via a form of ``emphatic'' voting was considered, taking into account the connections between nodes in addition to their opinions in the aggregation.  In \cite{salehi15}, two closely related problems were tackled: Preference inference and group recommendation based on partially observed ranking data. The main idea was to utilize the underlying social network structure under the assumption that the homophily and/or social influence shapes the network dynamics. The proposed models were tested empirically on several data sets, one of which was the Flixter data set \cite{jamali10}, consisting of a social network of movie watchers and their ratings of movies. Based on the ratings for each user, the relative number of movies of each genre watched by a specific user was calculated, and from these so-called \textit{user-genre scores}, ranking data of movie genres was constructed. However, due to the sparsity typical of movie ratings, only partial orders could be constructed in a this manner. Finally, in \cite{dhamal2018modeling}, two models were proposed for capturing how preferences are distributed among nodes in a typical social network. By sampling a small subset of representative nodes, the algorithms can harness the network structure to effectively construct an aggregate preference of the entire network population, and for preferences related to personal topics (such as lifestyle choices), the proposed approach was shown to be advantageous over traditional random polling. The said papers also connect to the (relatively rich) literature on preference aggregation and voting theory. However, network aspects seem to be rarely considered in that context, and we are only aware of \cite{salehi15,becirovic17,salehi19} and \cite{dhamal2018modeling}.

\subsubsection{Structural Balance}

The discovery of the structural balance theorem in 1956 \cite{cartwrightharary} has
spawned a large literature on empirical analysis of real-world
networks \cite{harary61,moore79,newcomb81,estrada2014social}, analysis of dynamic processes \cite{shi-altafini19,cisneros2019dynamic}, partially balanced networks \cite{aref17}, and perhaps most importantly in the current context, extensions to cases beyond the canonical signed-link \(+/-\) setup.  The most representative extensions are \cite{qian-extended-balance14} and \cite{meng20}. In \cite{qian-extended-balance14}, the edge weights can be any real number drawn from a total ordering. A distance metric is defined such that the negative (positive) are mapped to large (small) distances. A triangle is then said to be structurally balanced if the three distances involved satisfy the metric triangle inequality. In \cite{meng20}, the authors considered signed digraphs and redefine structural balance as a local node property: A node is called structurally balanced if a ceratin subgraph related to the node can be bipartitioned such that all directed edges within a partition have non-negative weights, and all directed edges between partitions have non-positive weights.

Another interesting direction of research is the generalization of structural balance in networks with node attributes. For example, in \cite{du21}, the authors defined balance in fully signed networks, i.e., networks where both the edges and the nodes have signs. Such a network is then said to be balanced if and only if it can be partitioned into clusters, within which nodes have identical attributes and all edges are negative. The authors provide an energy function whose minimum is taken as a measure of partial network balance, and they also propose an algorithm for the efficient computation of this value. In \cite{ he20} the method was further generalized to fully signed networks in which the number of attributes for each node is arbitrary (that is, not only \(+\) or \(-\)).

However, none of the existing literature on balance, to our knowledge, has dealt with networks with preference orderings as attributes.

\subsection{Contributions}

In this paper we discuss possible definitions of structural balance in a network with preference orderings as node attributes.  The main result (Theorem \ref{thm:main}) is that for the
three-alternative case (\(A,B,C\)) we reduce the \((3!)^3=216\) possible
configurations of triangles to 10 equivalence classes.  These 10
classes represent the 10 different types of triangles that can occur
in a network, and based on them, a notion of balance can be defined.
We also give a general formula for the number of equivalence classes
for the \(n\)-alternative case (Theorem \ref{thm:generalcase}).  Finally, we examine
numerically the data set in \cite{dhamal2018modeling} and compare its empirical
distribution of the equivalence classes to a null hypothesis in which
preference orderings are randomly assigned to the nodes.

	\section{Main results}\label{sec:model}
		We define a preference ordering on \(n\) alternatives, \(n \geq 2\),  as a permutation \(\sigma\) on \(n\) elements. A preference triangle on \(n\) alternatives, \(P_n\), is then defined to be a complete graph on \(3\) nodes (i.e., \(K_3\)), where each node is associated with a preference ordering. We introduce a relation \(\mathcal{R}\) and we say that two preference triads \(P_n^1\) and \(P_n^2\) are related if \(P_n^1\) can be transformed into \(P_n^2\) by relabeling its nodes and by applying the same permutation of the elements to all nodes (which corresponds to to relabeling the elements).  This is denoted by \(P_n^1\sim P_n^2\).
		
		For example, the two preference triads depicted in Figure~\ref{fig:preftriad} are related: In \(P_3^1\), change the node labels \(v_1, v_2, v_3\) to \(w_3, w_2, w_1\), respectively, and let elements \(A\) and \(B\) swap places in all three permutations to obtain \(P_3^2\).
		
		\begin{figure}[h]
		\centering
		\begin{tikzpicture}[main_node/.style={circle,fill=blue!20,draw,minimum size=1em,inner sep=3pt]}]
		
		\node[main_node] (1) at (0,0) [label=above:\(ABC\)] {\(v_1\)};
		\node[main_node] (2) at (-1, -1.5) [label=below:\(ACB\)]  {\(v_2\)};
		\node[main_node] (3) at (1, -1.5) [label=below:\(BAC\)] {\(v_3\)};
		\node (4) at (0,-2.7) {\(P_3^1\)};		
		\draw (1) -- (2) -- (3) -- (1);
		\end{tikzpicture}%
		\hspace{2em}
		\begin{tikzpicture}[main_node/.style={circle,fill=blue!20,draw,minimum size=1em,inner sep=3pt]}]
		
		\node[main_node] (1) at (0,0) [label=above:\(BAC\)] {\(w_3\)};
		\node[main_node] (2) at (-1, -1.5) [label=below:\(BCA\)]  {\(w_2\)};
		\node[main_node] (3) at (1, -1.5) [label=below:\(ABC\)] {\(w_1\)};
		\node (4) at (0,-2.7) {\(P_3^2\)};
		\draw (1) -- (2) -- (3) -- (1);
		\end{tikzpicture}
		\caption{Two  equivalent preference triads on \(3\) alternatives: \(P_3^1 \sim P_3^2\).}
		\label{fig:preftriad}
		\end{figure}
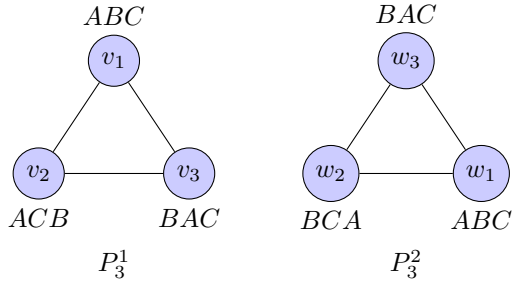
	It is easy to see that the relation \(\mathcal{R}\) is an equivalence relation: Let \(P_n^1, P_n^2\) and \(P_n^3\) be preference triads. Clearly \(P_n^1\sim P_n^1\) since no transformation is needed, so \(\mathcal{R}\) is reflexive. Furthermore, if \(P_n^1\sim P_n^2\), then we can recover \(P_n^1\) from \(P_n^2\) by reversing the swaps and undo the relabeling, so \(\mathcal{R}\) is symmetric. Finally, if \(P_n^1\sim P_n^2\) and \(P_n^2\sim P_n^3\) we can first transform \(P_n^1\) into \(P_n^2\), and then transform \(P_n^2\) into \(P_n^3\), so \(\mathcal{R}\) is transitive.
	
	This relation restricts the number of unique preference triangles. Our first theorem states that \(\mathcal{R}\) reduces the number of preference triangles on \(3\) alternatives to \(10\) cases.
	
 \begin{theorem}\label{thm:main}
The  \((3!)^3 = 216\) possible preference triads on \(3\) alternatives can be partitioned into \(10\) equivalence classes induced by the relation \(\mathcal{R}\).
 \end{theorem}
\begin{proof}
	 There are \((3!)^3 = 216\) possible preference triads on \(3\) alternatives, \(A,B,C\), but since we always can swap the alternatives such that one of the preferences is \(A\succ B\succ C\) (abbreviated \(ABC\)), the \(216\) cases can first be reduced to \((3!)^2 = 36\) cases. These \(36\) cases are listed in Table \ref{fig:equiv_classes}, where the three rows in each case represent the three nodes in the corresponding preference triad. The top row is always \(ABC\). Thus, a swap of two rows is equivalent to letting the two corresponding nodes change labels with each other. 
	
	The equivalence relation \(\mathcal{R}\) partitions the cases into different equivalence classes: Consider for example Case 2. We can list all of its possible transformations,
	\begin{equation*}
				\begin{smallmatrix}
			ABC & ABC & ABC\\
			ABC & ACB & ACB\\
			ACB & ABC & ACB,
		\end{smallmatrix}
	\end{equation*}
	
	\noindent
	where the middle matrix is obtained by swapping rows \(2\) and \(3\), and the last matrix is obtained by first swapping rows \(1\) and \(3\) and then swapping \(B\) and  \(C\) in all three rows. We identify the middle matrix as Case 7, and the last matrix as Case 8, so Case 2 is related to both of these cases under \(\mathcal{R}\). On the other hand, it is not related to any other case, for example Case \(23\). To see this, we can list all possible transformations of case \(23\), to obtain
	
	\begin{equation*}
		\begin{smallmatrix}
						ABC & ABC \\
			BCA & CAB \\
			CAB & BCA,
		\end{smallmatrix}
	\end{equation*}	
	\noindent
	and note that none of these matrices matches the transformations of Case 2. By proceeding in this manner for all \(36\) cases, it can be shown by exhaustion that there are exactly \(10\) equivalence classes, highlighted in color in Table \ref{fig:equiv_classes}. The equivalence classes with their representatives are listed in Table \ref{fig:equiv_partitioning}. 
\end{proof}

 \begin{table}
 	\addtocounter{table}{-1}
 		\centering
 		\caption{(a) The equivalence classes on \(3\) elements, and (b) their partitioning of the possible preference triads.}
 	\subfloat[The possible preference triads on \(3\) elements, with the \(10\) equivalence classes induced by the relation \(\mathcal{R}\) marked in color.]{
 		\includegraphics[scale=0.6]{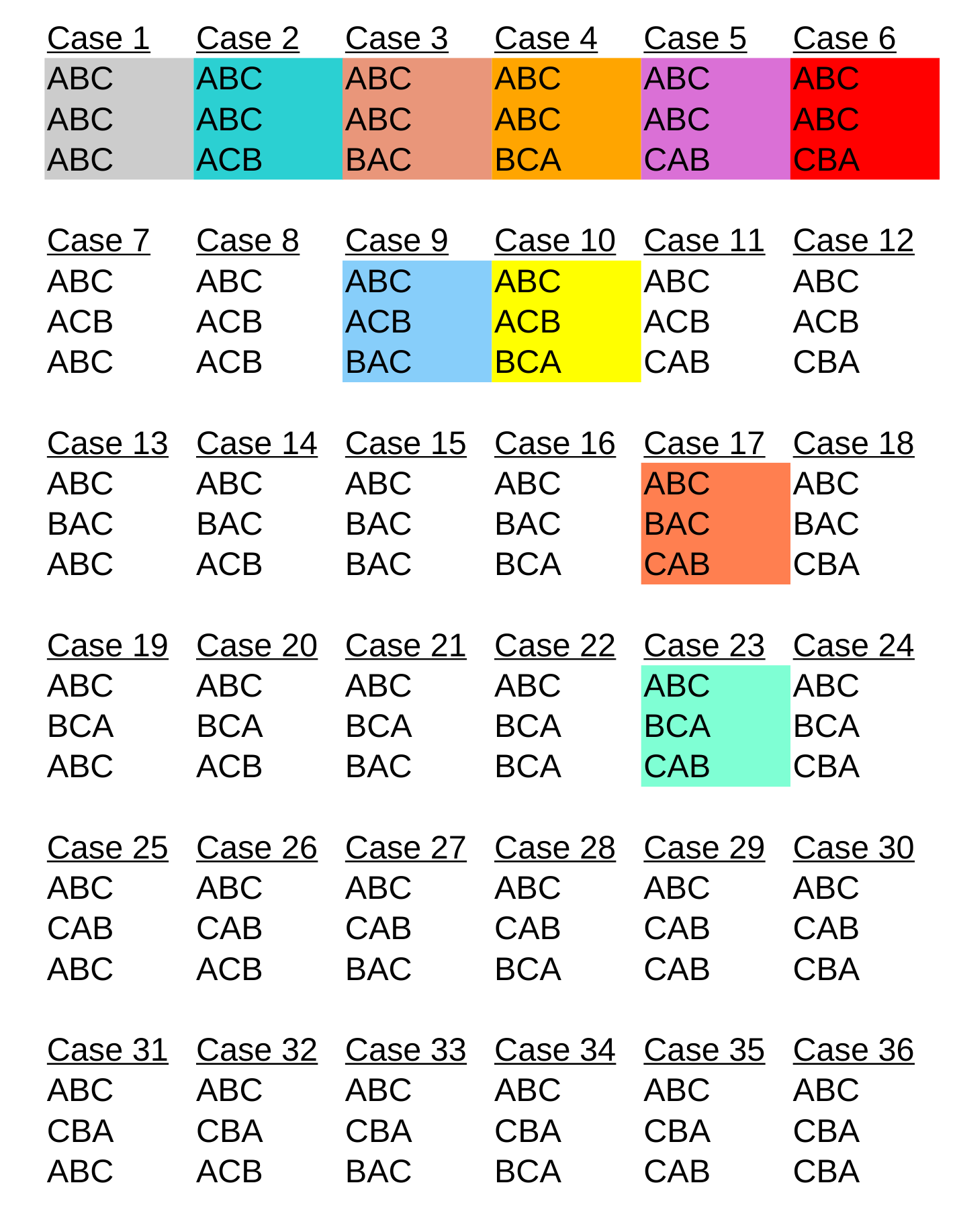}
 		\label{fig:equiv_classes}
	 	}
 	\hfil\subfloat[The\(10\)  equivalence classes and their representatives.]{
 		\begin{tabular}{|c|c|}
 			\hline
 			Equivalence class & Cases \\
 			\hline
 			1&1 \\
 			2&2, 7, 8\\
 			3&9, 11, 14, 16, 21,  26 \\
 			4&10, 12, 20, 30, 32, 35\\
 			5&3, 13, 15\\
 			6&17, 18, 24, 27, 33, 34\\
 			7&4, 19, 29\\
 			8&5, 22, 25\\
 			9&6, 31, 36\\
 			10&23, 28\\
 			\hline
 		\end{tabular}
 			\label{fig:equiv_partitioning}}
 \end{table}

A natural question is if it is possible to define a notion of a \textit{balanced} preference triangle. In classical balance theory, a
triangle is either balanced or unbalanced. There is no obvious analog to this idea for preference triangles since, by Theorem
\ref{thm:main}, each such triangle falls into one of \(10\) different equivalent classes. However, if the equivalence classes could be totally
ordered, it might be possible to define balance so that the order could be interpreted as ranging from ``most balanced'' to
``least balanced''. This is motivated by the fact that there seems to be at least one intuitive partial order on the equivalence
classes: In case \(1\) in Table \ref{fig:equiv_classes}, all nodes agree perfectly on the preference orderings, which could be interpreted as positive
relations between all three nodes, i.e., a \(+++\) triangle in terms of classical balance theory. The cases \(2, 3, 4, 5\) and \(6\)
(and their equivalents) might be interpreted as \(+--\) triangles, with further refinement of the internal order possible since in
case \(2\) all preference orderings starts with \(A\), while in cases \(3\) to \(6\) only two preference orderings start with \(A\). Similarly,
in cases \(9, 10, 17\) and \(23\) (and their equivalents), none of the preference orderings are identical, which might be interpreted
as \(---\) triangles, and again the internal order might be further refined since in case \(9\) and \(10\) two of the preferences
start with \(A\). In case \(23\) the preference orderings are in fact maximally different, so in some sense this could be seen as
the ``least balanced'' preference triangle. A triangle reminiscent of \(++-\) in classical balance theory does not exist. This is
an important observation, since it leaves only \(---\) as a possible unbalanced triangle. However, it has been argued (see,
e.g., \cite{davis67}) that the definition of balance should be generalized to permit also the triangle \(---\) as a balanced triangle, in
which case one talks about weak balance theory. Therefore, in view of the above mapping from preference triangles to signed
triangles, all triangles are weakly balanced.

Consequently, a notion of balance for preference triangles is perhaps best defined in terms of partial balance, a concept
discussed in \cite{aref17}, since the equivalence classes intuitively can be partially ordered. However, we have been unsuccessful in
finding a total order for the equivalence classes, and even though several partial orderings can be constructed if one
allows ties, we have not found an objective argument for preferring one partial ordering over another. Therefore we have
not addressed this question in detail in the current paper, and while we may not have a firm answer, it might be an interesting
direction to explore further.

Another natural question is if Theorem \ref{thm:main} can be generalized
to the case with \(n\) alternatives, with \(n \geq 2\). In the appendix we derive a closed form expression for the number of equivalence classes induced by \(\mathcal{R}\), and show that it can be expressed as a function of \(n\) (the number of elements in the permutations). In particular, we prove the following theorem.
	\begin{theorem}\label{thm:generalcase}
		Let \(n\in \mathbb{N}\) with \(n \geq 2\), let \(\mathcal{S}\) denote the set of all preference triads on \(n\) elements and let \([x]\) denote the integer part of \(x\in\mathbb{R}\). Then the number of equivalence classes of \(\mathcal{S}\) induced by the relation \(\mathcal{R}\) is
		\begin{equation*}
		\left|\mathcal{S}/\mathcal{R}\right| = \dfrac{n!(n!+3)+2(\ell_n+1)}{6},
		\end{equation*}
		where
	 \begin{equation}
		\ell_n = \sum_{m=1}^{\left[\dfrac{n}{3}\right]} \dfrac{n!}{(n-3m)!m!3^m}.
		\end{equation}
	\end{theorem}

	Since the number of equivalence classes increases super-exponentially with \(n\) (for \(n=4,5,6,\dots\) there are \(111,2467,86787,\dots\) number of equivalence classes, respectively), it quickly becomes impractical to explicitly list all of them. Therefore this paper is confined to the minimal non-trivial case, \(n = 3\).
	
\section{Experimental results}
In this section, we analyze an authentic social network with preferences as node attributes.  In the standard theory of structural balance, a complete signed network is said to be balanced if all of its triangles are balanced. We extend this idea to preference triangles and  enumerate all such triangles into the different equivalence classes described in the previous section. The aim is to determine if the distribution of equivalence classes of preference triangles is significantly different from what one would expect by chance: Specifically we construct synthetic networks by performing randomized degree-preserving edge-rewirings \cite[Chapter 4]{barabasi16} on the real-world network, where the nodes have preferences that follow the empirical distribution of the authentic data set. The null hypothesis is that there is no  significant difference between the distribution of equivalence classes for the authentic network and those of the synthetic networks.

The data set, shared with us by the authors of \cite{dhamal2018modeling},  was collected from a specially designed Facebook app where users were asked to rank their preferences on \(8\) topics, with each topic containing \(5\) items, see \cite[Table VI, Appendix A]{dhamal2018modeling} for details. Thus the data set consists of each individual's preferences on each topic together with the underlying social network (there is a link between two users if they are friends on Facebook). The network consists of \(844\) nodes, \(6129\) edges and the fraction of closed triangles is \(0.4542\).\cite[Table 1]{dhamal2018modeling}

The data is analyzed as follows: For each of the \(8\) topics, we pick a subset of \(3\) out of the \(5\) available items, giving us \(8\binom{5}{3} = 80\) sets of preferences, where each set contains \(844\) individual preferences (one per node). The internal order of these partial preferences is preserved, so for example if one of the original preferences is \(ADCEB\) and we select \(\{A,B,E\}\) as our subset of items, then the extracted preference becomes \(AEB\). For each of the \(80\) so-obtained preference sets, we calculate the empirical distribution of all possible preferences over the nodes. We then construct 10 synthetic networks by randomly rewiring the edges in the network such that the node degrees are preserved. For each of the \(80\) preference sets, we let the preferences of the nodes in the synthetic networks follow the same empirical distribution as the authentic data. Finally, we calculate the distribution of the \(10\) equivalence classes for both the authentic and synthetic networks. Thus we obtain \(80\) sets of real-world data, where each set is compared to \(10\) artificially constructed networks in order to test our null hypothesis.

In our analysis we were unable to find a sufficient significant difference to reject the null hypothesis, as illustrated by Figure \ref{fig:hist_pref}: The histograms of equivalence classes for the \(10\) synthetic networks and that of the authentic network overlap to a high extent, and we found similar results for all of the \(80\) preference sets. As noted in Section \ref{sec:intro}, research on preference orders as node attributes in a network is scarce, and we are currently not aware of any other available data sets. Therefore it is at this point unclear whether this negative result is due to the particular data set or if it is indicative of a more general phenomenon.

\begin{figure}[tp]
	\centering
		\includegraphics[scale=0.35]{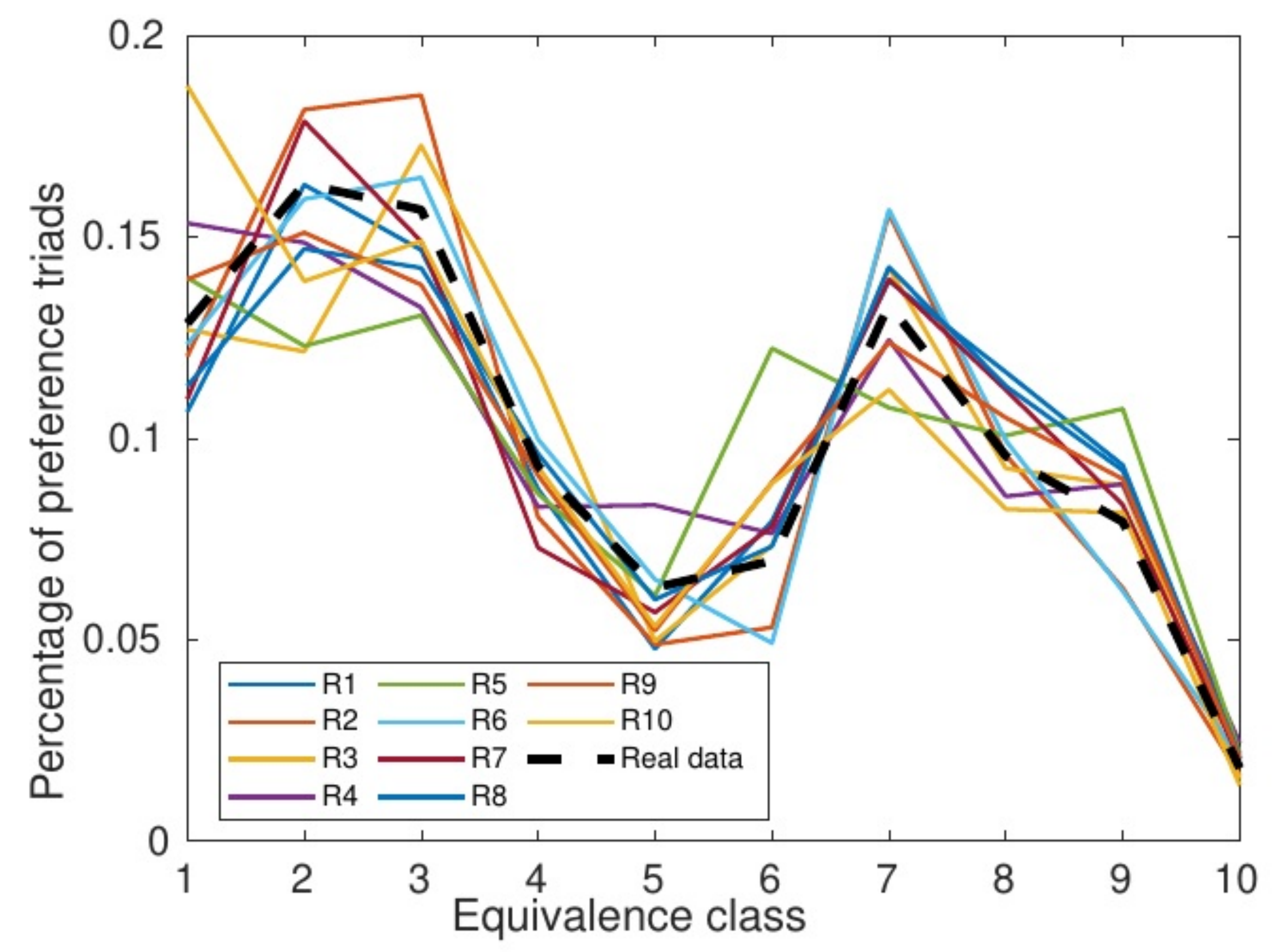}
		\caption{Histograms over equivalence classes of preference triads for the authentic network compared with \(10\) random degree-preserving edge-rewirings of said network (R1-R10), where the distribution of preferences in the rewired networks follow the empirical distribution of the authentic network. Here, the topic is ``Facebook activity'', and the alternatives are ``Viewing posts'', ``Chatting'' and ``Games/Apps'}
		\label{fig:hist_pref}
\end{figure}

\begin{center}
	\begin{table*}[t]
		\caption{Reproduction of Table VI in \cite[Appendix A]{dhamal2018modeling}}
		{\small
			\hfill{}
			\begin{tabular}{|c|c|c|c||c|c|c|c|}
				\hline
				\multicolumn{4}{|c||}{\textbf{Personal}}& \multicolumn{4}{|c|}{\textbf{Social}}\\
				\cline{3-8}
				\hline
				Hangout & Chatting & Facebook & Lifestyle & Website & Government & Serious & Leader \\
				place & app & activity & & visited & investment & crime & \\
				\hline
				Friend's place & WhatsApp & Viewing posts & Intellectual & Google & Education & Rape & N. Modi (India)\\
				Adventure park & Facebook & Chatting & Exercising & Facebook & Agriculture & Terrorism & B. Obama (USA)\\
				Trekking & Hangouts & Posting & Social activist & Youtube & Infrastructure & Murder & D. Cameron (UK)\\
				Mall & SMS & Games/Apps & Lavish & Wikipedia & Military & Corruption & V. Putin (Russia) \\
				Historical place & Skype & Marketing & Smoking & Amazon & Space explore & Extortion & X. Jinping (China)\\
				\hline
		\end{tabular}}
		\hfill{}
		\label{tab:origdata}
	\end{table*}
\end{center}

\section{Conclusions and discussion}
We have characterized the triangles that can occur in net-
works with preference orders as node attributes. Specifically,
for preference orders with three alternatives, we showed in
Theorem \ref{thm:main} that there are only \(10\) unique preference triangles.
We have also analyzed numerically an authentic data set
and compared its empirical distribution of unique preference
triangles with a null hypothesis with randomized preference
orderings. We hope that these results will stimulate others to
explore variations of the framework presented here, and collect
data that can be used for further quantitative studies. Some
open problems include the following:
\begin{itemize}
	\item  Is it possible to find a total order for the equivalence
	classes?
	\item Is there an objective argument for choosing a particular
	partial order for the equivalence classes? While we are not aware of any literature that addresses this particular issue, a recent paper \cite{gao21}
	proposed a partial order of the set of preference profiles between individuals.  Another paper of potential interest is \cite{jiang11}, where combinatorial Hodge theory was proposed as a tool for statistical analysis of ranking data through minimization of pairwise ranking disagreement. To the best of our knowledge, it is an open problem for both of these approaches whether or not they are generalizable to comparisons between groups of \(n \geq 2\) preferences per group (with \(n = 3\) being the special case of interest in our setting).
	\item Given an order of equivalence classes, how should the
	different classes map to structural balance? That is, how
	should such a mapping be formally defined?
	\item Does the distribution of equivalence classes differ significantly between real-world networks and synthetic
	networks? More samples of authentic networks with
	preferences as node attributes are needed for a robust
	analysis.
\end{itemize}
As a final remark, note that by pairing the two node attributes associated with a particular edge, a network with node attributes could always be mapped onto a network with edge
attributes. In particular, in an attempt to interpret structural balance in terms of node preferences, one could consider networks with two different types of nodes, representing two
different opinions (\(A\) and \(B\), say) and define an edge to be ``\(+\)'' if it connects two nodes of type \(A\) or two nodes of type \(B\), and ``\(-\)'' otherwise. Unfortunately, this does not lead 
anywhere as the resulting network is always balanced (in fact, it has the natural partitioning in two parts consisting of \(A\)-nodes and \(B\)-nodes, respectively). In order to map node 
attributes to edge attributes in a way that leads to non-trivial results, one must go beyond binary node attributes. The path we explored in this paper has been to associate preference 
orderings with the nodes. Future work may consider other possibilities.

\appendix
In order to prove Theorem \ref{thm:generalcase}, we first need a lemma.
\begin{lemma}\label{thm:lem}
	Let \(\ell_n\) denote the number of elements of order \(3\) in the symmetric group \(S_n\). Then
	\begin{equation}\label{eq:numel_ord3_Sn}
	\ell_n = \sum_{m=1}^{\left[\dfrac{n}{3}\right]} \dfrac{n!}{(n-3m)!m!3^m}.
	\end{equation}
\end{lemma}
\begin{proof}
	Any element \(g\in S_n\) can be written as a product of disjoint cyclic permutations, and the order of \(g\) is the least common multiple of the orders of these cycles. Thus \(g\) has order \(3\) only if its cyclic decomposition consists of identities and \(3\)-cycles. The latter will be constructed from a subset with \(3m\) elements, for some positive integer \(m\) such that \(3m\leq n\), and there are \(\binom{n}{3m}\) ways of choosing such a subset. From this set we can create \(m\) disjoint \(3\)-cycles in \((3m)!\) ways to obtain a product of the form
	\begin{equation}
	\underbrace{
		\underbrace{(a_2a_3a_3)}_\text{\(3\)-cycle} \underbrace{(a_4a_5a_6)}_\text{\(3\)-cycle}
		\dots
		\underbrace{(a_{3m-2}a_{3m-1}a_m)}_\text{\(3\)-cycle}}
	_\text{m \(3\)-cycles}
	\end{equation}
	Since these cycles are disjoint, there are \(m!\) ways to permute them. In turn, each \(3\)-cycle can be permuted in \(3!/2 = 3\) unique ways, giving us \(3^m\) ways to arrange them in total.
	
	Putting this together, we have the following result. For any positive integer \(m\leq [n/3]\), the number of elements of order \(3\) in \(S_n\) is equal to
	\begin{equation}
	\begin{aligned}
	&\binom{n}{3m} \dfrac{(3m)!}{m!3^m}\\
	&= \dfrac{n!(3m)!}{(3m)!(n-3m)!m!3^m}\\
	&= \dfrac{n!}{(n-3m)!m!3^m}
	\end{aligned}
	\end{equation}
	The result follows by summing over all \(m\) such that\\ \(1 \leq m \leq [n/3]\).
\end{proof}
Now we proceed to prove the Theorem \ref{thm:generalcase}.
\begin{proof}[Proof of Theorem \ref{thm:generalcase}]
	Note that we can always relabel the nodes and swap elements in the permutations such that the permutation associated with one of the nodes is the identity permutation, denoted by \(\varepsilon\). Therefore we only need to consider ordered lists of the form \((\varepsilon,\sigma,\pi)\), where \(\sigma\) and \(\pi\) are arbitrary permutations on \(n\) elements. Four cases can occur:
	\begin{subequations}\label{eq:cases}
		\begin{align}
		&(\varepsilon,\varepsilon,\varepsilon) \label{eq:case1} \\
		&(\varepsilon,\varepsilon,\sigma)\sim (\varepsilon,\sigma^{-1},\sigma^{-1})\sim (\varepsilon,\sigma,\varepsilon), \ \sigma \neq \varepsilon \label{eq:case2}  \\
		&(\varepsilon,\sigma,\sigma^2)\sim (\varepsilon,\sigma^2,\sigma), \ 
		\begin{cases}
		\sigma \neq \varepsilon \\
		\sigma^3 = \varepsilon
		\end{cases} \label{eq:case3} \\
		& (\varepsilon,\sigma,\sigma\pi)\sim (\varepsilon,\sigma\pi,\sigma) \label{eq:case4} \\
		&\sim (\varepsilon,\sigma^{-1},\pi) \sim (\varepsilon,\pi,\sigma^{-1}) \nonumber\\
		&\sim (\epsilon,\pi^{-1},\sigma^{-1}\pi^{-1}) \sim (\varepsilon,\sigma^{-1}\pi^{-1},\sigma^{-1}), \ 
		\begin{cases}
		\sigma \neq \varepsilon \\
		\pi \neq \varepsilon \\
		\sigma \neq \pi.
		\end{cases} \nonumber
		\end{align}
	\end{subequations}
	In cases \eqref{eq:case2} to \eqref{eq:case4} we can obtain equivalent lists by multiplying the permutations with an inverse permutation that takes one of them to the identity (e.g. \((\varepsilon,\varepsilon,\sigma)\) can be multiplied with \(\sigma^{-1}\) resulting in \((\varepsilon,\sigma^{-1},\sigma^{-1})\)). By relabeling the nodes we can also obtain additional equivalent lists: \((\varepsilon,\varepsilon,\sigma)\) is equivalent to \((\varepsilon,\sigma,\varepsilon)\). By doing so we get \(1, 3, 2\) and \(6\) possibilities for cases \eqref{eq:case1}, \eqref{eq:case2}, \eqref{eq:case3} and \eqref{eq:case4}, respectively.
	
	The number of equivalence classes will be equal to the sum of the number of unique representatives in each case. In case \eqref{eq:case1} there is only one possibility. In case \eqref{eq:case2} there are \(n!-1\) possibilities since \(\sigma\) can be any permutation on \(n\) symbols except the identity. In \eqref{eq:case3} we require \(\sigma\) to be of order \(3\), and by Lemma \ref{thm:lem} the number of possibilities for such permutations is \(\ell_n\). It follows that the number of possibilities is \(\ell_n/2\). Finally, note that the number of ways to arrange \((\varepsilon, \sigma, \pi)\) is equal to \((n!)^2\). Therefore we can deduce that the number of possibilities for \eqref{eq:case4} must be equal to
	\begin{equation}
	\dfrac{(n!)^2-1-3(n!-1)-2\ell_n/2}{6}.
	\end{equation}
	Thus the total number of equivalence classes is
	\begin{equation}
	\begin{aligned}
	\left|\mathcal{S}/\mathcal{R}\right| &= 1 + (n!-1) + \dfrac{\ell_n}{2} + \dfrac{(n!)^2-1-3(n!-1)-2\ell_n/2}{6}\\
	&= \dfrac{n!(n!+3)+\ell_n/2}{6}.
	\end{aligned}
	\end{equation}
\end{proof}

\section*{Acknowledgments}
We thank the authors of \cite{dhamal2018modeling} for generously sharing the datset that we used for the numerical experiments.
	\bibliographystyle{IEEEtran}
	\bibliography{preftriads_equivclasses}
\end{document}